\documentclass[11pt,a4paper,reqno]{amsart}
\usepackage[utf8]{inputenc}
\usepackage[english]{babel}
\usepackage{amsmath}
\usepackage{amsthm}
\usepackage{mathtools}
\usepackage{enumitem}
\usepackage{amsfonts}
\usepackage{amssymb}
\usepackage{bbm}
\usepackage{tikz}
\usetikzlibrary{matrix,arrows,chains,positioning,scopes,cd}
\usepackage[hidelinks]{hyperref}

\author{Carina Boyallian and Juan Guzm\'an}
\email{boyallia@mate.uncor.edu - juan.gabriel.guzman@unc.edu.ar}
\address{Facultad de Matemática, Astronomía, Física y Computación, Universidad Nacional de Córdoba. CIEM-CONICET, Medina Allende s/n (5000) Ciudad Universitaria, Córdoba, Argentina}

\title{Formal vertex laws associated to Lie conformal algebras}
\date{}

\theoremstyle{plain}

\newtheorem{theorem}{Theorem}[section]
\newtheorem{lemma}[theorem]{Lemma}
\newtheorem{corollary}[theorem]{Corollary}
\newtheorem{proposition}[theorem]{Proposition}

\theoremstyle{definition}

\newtheorem{definition}[theorem]{Definition}
\newtheorem{example}[theorem]{Example}

\theoremstyle{remark}
\newtheorem{remark}[theorem]{Remark}
\newtheorem*{claim*}{Claim}

\numberwithin{equation}{section}

\newcommand{\C}{\mathbb{C}}

\newcommand{\Z}{\mathbb{Z}}
\newcommand{\N}{\mathbb{N}}
\newcommand{\Del}{\partial}
\newcommand{\dst}{\displaystyle}
\newcommand{\lam}{\lambda}

\newcommand{\U}{\mathcal{U}}
\newcommand{\1}{\textbf{1}}
\newcommand{\Ydual}{\rotatebox[origin=c]{180}{Y}}
\newcommand{\hotimes}{\mathbin{\hat{\otimes}}}

\newcommand\blfootnote[1]{%
  \begingroup
  \renewcommand\thefootnote{}\footnote{#1}%
  \addtocounter{footnote}{-1}%
  \endgroup
}

\newsavebox{\tempbox}

\begin{document}

\begin{abstract}
We introduce several definitions within the framework of vertex and conformal algebras which are analogous to some important concepts of the classical Lie theory. Most importantly, we define formal vertex laws, which correspond to the notion of formal group law. We prove suitable vertex/conformal versions of a number of classical results such as the Milnor-Moore theorem, Cartier duality, and the equivalence between formal group laws and Lie algebras. 
\end{abstract}

\maketitle

\section{Introduction}

\blfootnote{This work was partially supported by CONICET and SeCyT.}
Vertex algebras were first defined (under the name ``chiral algebras'') by Borcherds in \cite{Borcherds} as a purely algebraic description of the chiral part of a conformal field theory (CFT). The OPE of the chiral fields in a CFT form an algebraic structure that in a certain way resembles a Lie algebra, which was named Lie conformal algebra \cite{Kac_beginners}. There are strong parallelisms between vertex algebras and associative algebras, as well as between Lie conformal algebras and Lie algebras \cite{Beilinson_Drinfeld,Kac_beginners}. The most important analogy may be the fact that any Lie conformal algebra $R$ has a universal enveloping vertex algebra $\U(R)$, which enjoys many of the properties that universal enveloping associative algebras of Lie algebras have \cite{Kac_beginners,Bakalov_field_algebras}. 

This relationship suggests that it might be worth studying other constructions arising in Lie theory in order to find out whether they can be assigned some object in the conformal setting having similar properties. In this paper we follow this path and present conformal/vertex analogues for several important notions and classical results, amongst which we have bialgebras, Cartier duality, and most importantly, formal group laws.

The classical Lie theory establishes certain equivalences of categories, which we will now briefly recount. Given a Lie algebra $\mathfrak{g}$, we have that $\U(\mathfrak{g})$ is a cocommutative and connected bialgebra, and all such bialgebras are of that form due to the Milnor-Moore theorem \cite{Montgomery}. Its dual space $\U(\mathfrak{g})^*$, endowed with the linearly compact topology, becomes a commutative and local topological bialgebra through Cartier duality \cite{Serre,Dieudonne,Hazewinkel_formalgroups}. Moreover, if $\mathfrak{g}$ has a basis indexed by $\mathcal{I}$, the PBW basis of $\U(\mathfrak{g})$ can be used to define an algebra isomorphism $\U(\mathfrak{g})^* \simeq \C[[X_{i}: i \in \mathcal{I}]]$. Using this map to translate the coalgebra structure of $\U(\mathfrak{g})^*$ to this algebra, we may define the formal power series $F_{l}(X,Y)=\Delta(X_{l}) \in \C[[X_{i},Y_{i} : i \in \mathcal{I}]]$ for each $l \in \mathcal{I}$. 

The set $F(X,Y)=\{F_{l}(X,Y) \}_{l \in \mathcal{I}}$ satisfies the following properties
\begin{enumerate}[label=(\roman*)]
    \item $F(X,0)=X=F(0,X)$.
    \item $F(X,F(Y,Z))=F(F(X,Y),Z)$.
\end{enumerate}
Any $F(X,Y)$ satisfying these axioms receives the name \textit{formal group law} (some caution needs to be taken when $\mathcal{I}$ is infinite, we will specify the details in Section \ref{section_formal_vertex_laws}). It can be proved \cite{Hazewinkel_formalgroups} that there exists a formal power series $I(X)$ such that $F(X,I(X))=F(I(X),X)=0$, and thus a formal group law may be interpreted as a formalization of the axioms that define a group structure.

The correspondences described in the above paragraphs may be summarized as the following chain of equivalences of categories.
\vspace{-0.3cm}
\begin{center}
\begin{tikzpicture}
    \draw (0,0) node {$\bigg\{$};
    \draw (0.9,0.25) node {Lie};
    \draw (0.9,-0.25) node {algebras};
    \draw (1.8,0) node {$\bigg\}$};

    \draw (2.2,0) node {$\simeq$};

    \draw (2.6,0) node {$\bigg\{ $};
    \draw (4.1,0.25) node {Cocomm. conn.};
    \draw (4.1,-0.25) node {bialgebras};
    \draw (5.6,0) node {$\bigg\}$};

    \draw (6,0) node {$\simeq$};

    \draw (6.4,0) node {$\bigg\{$};
    \draw (7.8,0.25) node {Comm. local};
    \draw (7.8,-0.25) node {top. bialgebras};
    \draw (9.2,0) node {$\bigg\}$};
    \draw (9.45,0.3) node {\textsuperscript{op}};

    \draw (9.75,0) node {$\simeq$};

    \draw (10.15,0) node {$\bigg\{$};
    \draw (11.2,0.25) node {Formal};
    \draw (11.2,-0.25) node {group laws};
    \draw (12.25,0) node {$\bigg\}$};
\end{tikzpicture}
\end{center}

This work is dedicated to the task of reproducing these equivalences within the vertex/conformal context. Given a Lie conformal algebra $R$, we prove an appropriate version of the Milnor-Moore theorem between $R$ and the vertex bialgebra $\U(R)$ (first defined in \cite{Li_bialgebras}), and then we define \textit{vertex bicoalgebras} as the suitable structure to consider in $\U(R)^*$ in order to establish a form of Cartier duality. The notion of duality for vertex algebras has been studied in \cite{Hubbard_vertex_coalgebras,Hubbard_vertex_operator_coalg}, but our approach is different because we need to take into account the topology of these spaces. Lastly, choosing a vector space basis of $R$ indexed by $\mathcal{I}$, we give an algebra isomorphism $\U(R)^{*}  \simeq \C[[X_{i}: i \in \mathcal{I}]]$ and use it to encode the additional ``vertex structure'' of $\U(R)^*$ into a set of formal power series $$F(x)(X,Y)=\Big\{ \sum_{n \in \Z} F^{n}_{l}(X,Y) x^{-n-1}  \Big\}_{l \in \mathcal{I}} \in \C[[X_{i},Y_{i},x,x^{-1}]]^{\mathcal{I}},$$
which we call \textit{formal vertex laws}, since the vertex bicoalgebra axioms in $\U(R)^{*}$ are translated into certain properties of these formal power series which are similar to the axioms of vertex algebras.

Therefore, we have been able to find analogies for most of the classical concepts that we described earlier. One exception though consists of the antipode for $\U(R)$, preventing us from defining the notion of Hopf vertex algebra. The main issues arise from the fact that $\U(R)$ is not an associative algebra. Nonetheless, this does not prevent us from pursuing our goals, since in the classical theory the antipode of a cocommutative connected Hopf algebra is not an additional piece of information, it is determined by its bialgebra structure.

This work is part of an ongoing research whose aim is to define a conformal analogue for the notion of Lie group. Lie groups may be studied at the infinitesimal, local or global level; we present here the conformal version of the infinitesimal level of the theory and leave the local and global parts for future publications.

The first author would like to thank V. Kac for suggesting this problem.

\tableofcontents

\section{Preliminaries}
We begin by introducing the basic definitions. All vector spaces will be over $\C$. We will consider commuting formal variables $x$, $x_0$, $x_1$ and $x_2$. Let the ``formal $\delta$-function'' be $$\delta(x)=\sum_{n \in \Z}x^n.$$

For any $n \in \Z$, we define $$(x_1 \pm x_2)^n = \sum_{k \geq 0} \binom{n}{k} x_1^{n-k}(\pm x_2)^k,$$ where $\binom{n}{k}=\frac{n(n-1) \cdots (n-k)}{k!}$ for $n \in \Z$ and $k \in \Z_{+}=\N \cup \{0\}$. We can present now the definition of vertex algebra that we shall work with.

\begin{definition}
A \textit{vertex algebra} consists of a vector space $U$ together with a distinguished element $\1 \in U$ called \textit{vacuum} vector and a linear map
\begin{align*}
Y(\cdot,x) : U &\rightarrow (\text{End }U)[[x,x^{-1}]] \\
             u &\mapsto Y(u,x)=\sum_{n \in \Z}u_{(n)}\,x^{-n-1},
\end{align*}
satisfying the following axioms:
\begin{enumerate}[label=(\roman*)]
\item Truncation: For all $u,v \in U$,
\begin{equation}\label{truncation_vertex}
Y(u,x)v \in U((x)).
\end{equation}
Here $U((x))$ is the space of Laurent power series in $x$ with coefficients in $U$. This axiom is equivalent to requiring $u_{(n)}v=0$ for $n >> 0$.
\item Left unit: For all $u \in U$,
\begin{equation}\label{left_unit_vertex}
Y(\1,x)u=u.
\end{equation}
\item Creation: For all $u \in U$, $Y(u,x)\1$ is a holomorphic formal power series in $x$ and
\begin{equation}
Y(u,x)\1 \big|_{x=0} = u.\label{creation_vertex}
\end{equation}
Note that (\ref{creation_vertex}) means that the independent term of $Y(u,x)\1$ is $u$.
\item Jacobi identity: For all $u,v \in U$,
\begin{align}
 x_0^{-1}\delta\left(\frac{x_1-x_2}{x_0}\right)&Y(u,x_1)Y(v,x_2) - x_0^{-1}\delta\left(\frac{x_2-x_1}{-x_0}\right)Y(v,x_2)Y(u,x_1) \nonumber\\
&=x_2^{-1}\delta\left(\frac{x_1-x_0}{x_2}\right)Y(Y(u,x_0)v,x_2).
\end{align}
\end{enumerate}
\end{definition}

\begin{remark}
This is the definition used in \cite{Lepowsky_Li} and \cite{Hubbard_vertex_coalgebras}, but there exist several equivalent definitions (cf. \cite{Kac_beginners,Bakalov_field_algebras}).
\end{remark}

The truncation axiom is vital in this definition, since in order for the Jacobi identity to make sense, we must check that the coefficients of each monomial in the three formal variables become finite sums when applied to elements in $U$, and this is true because of the truncation property.

The operations $\cdot_{(n)}\cdot : U \otimes U \rightarrow U$, $u \otimes v \mapsto u_{(n)}v$ for $n \in \Z$ are called the \textit{$n$-th products} of the vertex algebra. The $(-1)$-st product is called \textit{normally ordered product}, and we denote it by $a_{(-1)}b=\,:ab:$. This product is generally non-associative, so whenever we take the normally ordered product of more than two elements it must be read associating the terms from right to left.

We define a linear operator $\Del: U \rightarrow U$ by $\Del u = u_{(-2)}\1$, called the \textit{translation operator}, which turns $U$ into a $\C[\Del]$-module. Setting $\Del^{(j)}=\frac{1}{j!}\Del^{j}$ for all $j \geq 0$, it is easy to prove that for any $u, v \in U$ the identity $$u_{(-j-1)}v=\,:(\Del^{(j)}u)v:$$holds, so the negative $n$-th products in $U$ are determined by the normally ordered product and the action of $\C[\Del]$. In particular, for all $u \in U$ we have $$Y(u,x)\1=e^{x\Del}u.$$
\begin{example}
If $U$ is any associative commutative algebra, then we can turn it into a vertex algebra by taking $\1$ to be the unit of the algebra, and taking all of its $n$-th products to be zero except for the normally ordered product, which we define as the product of the algebra. These are the most basic examples of vertex algebras, and they are called \textit{trivial vertex algebras}.
\end{example}

\begin{definition}
A \textit{vertex algebra homomorphism} is a linear map $f : U \rightarrow \tilde{U}$ between vertex algebras such that for all $u,u' \in U$, $$f(\1_{U})=\1_{\tilde{U}} \quad \text{and} \quad f(Y_{U}(u,x)u')=Y_{\tilde{U}}(f(u),x)f(u').$$
\end{definition}
Let us denote by $\mathcal{VA}$ the category whose objects are vertex algebras and whose morphisms are vertex algebra homomorphisms. This is a monoidal category, where the unit object is $\C$ with the trivial vertex algebra structure, and the monoidal product of two vertex algebras $U_1$ and $U_2$ is the tensor product $U_1 \otimes U_2$, with vertex algebra structure defined by taking $\1_{U_1} \otimes \1_{U_2}$ as the vacuum vector and $Y_{U_1 \otimes U_2}(u_1 \otimes u_2,x)=Y_{U_1}(u_1,x) \otimes Y_{U_2}(u_2,x)$ for $u_1, u_2 \in U$ (see \cite{Kac_beginners}).

We will focus our attention on a special type of vertex algebras: those that can be realized as universal enveloping algebras of Lie conformal algebras. Let us briefly recall this construction, as well as the definitions needed.

\begin{definition}
A \textit{Lie conformal algebra} is a $\C[\Del]$-module $R$ with a linear map
\begin{align*}
[\cdot_{\lam}\cdot]: R \otimes R &\rightarrow R [\lam] \\
                    a \otimes b &\mapsto [a_{\lam}b]=\sum_{n \in \Z_{+}}\frac{1}{n!}\, \lam^{n}\,a_{(n)}b
\end{align*}
such that the following conditions hold:
\begin{enumerate}[label=(\roman*)]
\item $\Del$-sesquilinearity: For all $a$, $b \in R$,
\begin{align}
[\Del a_{\lam}b] &= - \lam [a_{\lam}b] \text{ and} \\
[a_{\lam}\Del b] &= (\Del+\lam) [a_{\lam}b].
\end{align}
\item Antisimmetry: For all $a$, $b\in R$,
\begin{equation}
[a_{\lam}b]=-[b_{-\lam-\Del}a].
\end{equation}
\item Jacobi identity: For all $a$, $b$ and $c \in R$,
\begin{equation}\label{Jacobi_conformal}
[a_{\lam}[b_{\mu}c]] = [[a_{\lam}b]_{\lam+\mu}c] + [b_{\mu}[a_{\lam}c]].
\end{equation}
\end{enumerate}
\end{definition}

Once again we encounter $n$-th products, but in the case of Lie conformal algebras they are only defined for $n \in \Z_+$. In fact, a vertex algebra $U$ is always a Lie conformal algebra, when we take into account only the non-negative $n$-th products and define the $\C[\Del]$-action as $\Del a = a_{(-2)}\1$ for all $a \in U$. Moreover, vertex algebras may be equivalently defined as Lie conformal algebras with an extra structure of a $\C[\Del]$-differential unital non-associative algebra (given by the normally ordered product), satisfying certain compatibilities (cf. \cite{Bakalov_field_algebras,DeSole}). 

Let $R$ be a Lie conformal algebra. The bracket on $R$ defined by
\begin{equation}\label{Lie_bracket_R}
[a,b]=\int_{-\Del}^{0}[a_{\lam}b]\,d\lam \qquad \forall a,b \in R
\end{equation}
turns $R$ into a Lie algebra, which we denote by $R_{Lie}$. Thus we may construct its universal enveloping (associative) algebra $\U(R_{Lie})$, and it turns out that there exists a unique structure of a vertex algebra on $\U(R_{Lie})$ such that the restriction of its $\lam$-bracket to $R_{Lie} \times R_{Lie}$ coincides with the $\lam$-bracket of $R$ and the restriction of its normally ordered product to $R_{Lie} \times \U(R_{Lie})$ coincides with the associative product in $\U(R_{Lie})$ \cite[Thm.~7.12]{Bakalov_field_algebras}.

\begin{definition}
$\U(R_{Lie})$ with this vertex algebra structure is the \textit{universal enveloping vertex algebra} of the Lie conformal algebra $R$, and we denote it by $\U(R)$.
\end{definition}

The vertex algebra $\U(R)$ satisfies the following universal property (see \cite{Bakalov_field_algebras, DeSole}): if $W$ is a vertex algebra and $f: R \rightarrow W$ is a homomorphism of Lie conformal algebras (that is, a morphism of $\C[\Del]$-modules preserving the $\lam$-brackets), then there exists a unique homomorphism of vertex algebras $\tilde{f}: \U(R) \rightarrow W$ that extends $f$. That is, $\tilde{f} i = f$, where $i : R \rightarrow \U(R)$ is the inclusion map.

Since $\U(R)$ as a vector space is just the universal enveloping algebra of $R_{Lie}$, we have a PBW basis for $\U(R)$, and we may write it using the normally ordered product: if $\{a_{i} : i \in \mathcal{I} \}$ is an ordered basis for $R$ as a vector space, then we can form a basis of $\U(R)$ by considering all elements of the form $$:a_{i_1} \cdots \,a_{i_{k}}:,$$ where $i_1 \leq ... \leq i_k$ in $\mathcal{I}$ and $k \geq 0$.

\section{Vertex bialgebras and the Milnor-Moore theorem}\label{section_vertex_bialgebras}

Let us recall that for any Lie algebra $\mathfrak{g}$, there is a canonical structure of cocommutative bialgebra on $\U(\mathfrak{g})$, with coproduct and counit induced by $$\Delta(x)= 1 \otimes x + x \otimes 1,\quad \epsilon(x) = 0 \qquad \forall x \in \mathfrak{g}.$$

We will use Sweedler's notation $\Delta(x)=x_{(1)}\otimes x_{(2)}$ for coalgebras, where the summation is implicit.

A similar procedure can be performed on the universal enveloping vertex algebra $\U(R)$ of a Lie conformal algebra $R$: there are unique vertex algebra homomorphisms $\Delta: \U(R) \rightarrow \U(R) \otimes \U(R)$ and $\epsilon: \U(R) \rightarrow \C$ such that $\Delta(x)=\1 \otimes x + x \otimes \1$ and $\epsilon(x)=0$ for all $x \in R$ \cite{Li_bialgebras}.

Therefore, $\U(R)$ becomes our first example of a vertex bialgebra.

\begin{definition}(\cite{Li_bialgebras})
A \textit{vertex bialgebra} is a vertex algebra $U$ together with a coalgebra structure $(U,\Delta,\epsilon)$ such that the maps $$\Delta: U \rightarrow U \otimes U \quad \text{and} \quad \epsilon: U \rightarrow \C $$ are vertex algebra homomorphisms. Thus a vertex bialgebra is simply a coalgebra in the monoidal category of vertex algebras $\mathcal{VA}$.
\end{definition}

If $U$ is a vertex bialgebra, we shall denote by $\mathcal{P}(U)$ the subset of \textit{primitive} elements of $U$, that is, the elements $u \in U$ such that $\Delta(u)=u \otimes \1 + \1 \otimes u$. In the classical Lie theory the primitives of a bialgebra form a Lie algebra, and in our case we can prove a similar result.

\begin{proposition}
$\mathcal{P}(U)$ is a Lie conformal subalgebra of $U$.
\end{proposition}

\begin{proof}
Let $a, b \in \mathcal{P}(U)$. Then
\begin{align*}
\Delta(Y(a,x)b)&=Y(\Delta a, x)\Delta b\\
&=Y(a \otimes \1 + \1 \otimes a, x)(b \otimes \1 + \1 \otimes b)\\
&=Y(a,x)b \otimes Y(\1,x)\1 + Y(a,x)\1 \otimes Y(\1,x)b \\
&\quad + Y(\1,x)b \otimes Y(a,x)\1 + Y(\1,x)\1 \otimes Y(a,x)b\\
&=Y(a,x)b \otimes \1 + e^{x\Del} a\otimes b + b \otimes e^{x \Del}a + \1 \otimes Y(a,x)b \1.
\end{align*}
Since $e^{x \Del}a$ is holomorphic, comparing the coefficients of $x^{-n-1}$ in both sides for non-negative values of $n$ gives the formula $$\Delta(a_{(n)}b)=a_{(n)}b \otimes \1 + \1 \otimes a_{(n)}b,$$
and thus $a_{(n)}b \in \mathcal{P}(U)$ for all $n \geq 0$. Analogously, for any $a \in \mathcal{P}(U)$ we have
\begin{align*}
\Delta (Y(a,x)\1) &= Y(\Delta a,x)\Delta \1\\
&=Y(a,x)\1 \otimes \1 + \1 \otimes Y(a,x)\1,
\end{align*}
so comparing the coefficient of degree one on each side yields that $\Delta(\Del a)=\Del a \otimes 1 + \1 \otimes \Del a$, which means that $\mathcal{P}(U)$ is a $\C[\Del]$-submodule of $U$.
\end{proof}

\begin{remark}\label{PUR_equals_R}
If $R$ is a Lie conformal algebra, it is clear that $\mathcal{P}(\U(R)) = R$, because the primitives depend only on the coalgebra structure and the unit of $\U(R)$, and therefore $$\mathcal{P}(\U(R))=\mathcal{P}(\U(R_{Lie}))=R_{Lie}=R.$$
\end{remark}

A \textit{homomorphism of vertex bialgebras} is a map between vertex bialgebras which is both a homomorphism of vertex algebras and of coalgebras. We can extend the universal property of $\U(R)$ to the vertex bialgebra setting.

\begin{proposition}\label{universal_property_vertex_bialgebras}
Let $W$ be a vertex bialgebra and $\phi: R \rightarrow \mathcal{P}(W)$ a homomorphism of Lie conformal algebras. Then there exists a unique homomorphism of vertex bialgebras $\tilde{\phi}: \U(R) \rightarrow W$ that extends $\phi$.
\end{proposition}

\begin{proof}
We only need to check that the induced vertex algebra homomorphism $\tilde{\phi}:\U(R) \rightarrow W$ given by the universal property of $\U(R)$ as a vertex algebra is also a coalgebra map.

Clearly $\Delta_{W}(\tilde{\phi}(\1))=(\tilde{\phi} \otimes \tilde{\phi})\Delta(\1)$. If $a \in R$, then $\tilde{\phi}(a)=\phi(a)$ is primitive in $W$ by hypothesis, and thus $\Delta_{W}(\tilde{\phi}(a))=\phi(a) \otimes \1 + \1 \otimes \phi(a) = (\phi \otimes \phi)\Delta(a)$.

Now let $a=\,:a_1 \cdots a_n:\, \in \U(R)$, where $a_i\in R$ for $i=1,...,n$ and $n \geq 1$. We have that
\begin{align*}
\Delta_{W}(\tilde{\phi}(a))&=\Delta_W(:\phi(a_1) \cdots \phi(a_n):)\\
&=\,:\Delta_W(\phi(a_1)) \cdots \Delta_W(\phi(a_n)):\\
&=\,:(\phi \otimes \phi)(\Delta a_1) \cdots (\phi \otimes \phi)(\Delta a_n):\\
&=(\tilde{\phi} \otimes \tilde{\phi})\Delta (a).
\end{align*}

Therefore $\tilde{\phi}$ preserves the coproduct, and analogous computations show that it also preserves the counit.
\end{proof}

Let us recall that a Lie conformal algebra $R$ may be seen as a Lie algebra $R_{Lie}$. By construction, $\U(R)$ coincides with $\U(R_{Lie})$ not only as a vector space but also as a coalgebra. In particular, $\U(R)$ is \textit{connected} (that is, its coradical is $\C \1$) and cocommutative. The classical Milnor-Moore theorem states that these conditions are necessary and sufficient for a bialgebra over a field of characteristic zero to be the universal enveloping algebra of its Lie algebra of primitives. Now we shall prove the analogous statement for vertex bialgebras.

\begin{theorem}[Milnor-Moore for vertex bialgebras]\label{Milnor_Moore}
If $U$ is a connected cocommutative vertex bialgebra, then $$U \simeq \U(\mathcal{P}(U))$$as vertex bialgebras.
\end{theorem}

\begin{proof}
Let $R=\mathcal{P}(U)$. Using the universal property proved in Proposition \ref{universal_property_vertex_bialgebras}, we may lift the identity morphism $R \rightarrow \mathcal{P}(U)$ to a vertex bialgebra map $$\varphi: \U(R) \rightarrow U.$$
The injectivity of $\varphi$ follows from \cite[Lemma~5.3.3]{Montgomery}, which states that whenever $f:C \rightarrow D$ is a coalgebra map with $C$ connected and $f|_{\mathcal{P}(C)}$ injective, then $f$ must be injective. On the other hand, its surjectivity can be obtained following verbatim the proof of the classical Milnor-Moore theorem given in \cite[Thm.~5.6.5]{Montgomery}, since it only uses the coalgebra structure of $\U(R)$ and its PBW basis, which is also the same as in $\U(R_{Lie})$.
\end{proof}

It is easy to check that both $\U$ and $\mathcal{P}$ are actually functorial, so the results obtained in Remark \ref{PUR_equals_R} and Theorem \ref{Milnor_Moore} may be summarized categorically.

\begin{proposition}\label{equiv_conformal_and_vertex_bialg}
The functors $\U$ and $\mathcal{P}$ define an equivalence of categories between the category of Lie conformal algebras and the category of connected cocommutative vertex bialgebras.
\end{proposition}


\section{Topological vertex coalgebras}

The next step in our program is to determine the correct structure to consider in the dual space $\U(R)^*$ of the universal enveloping vertex algebra of a Lie conformal algebra $R$. In order to do that, in this section we shall study the dual spaces of arbitrary vertex algebras, which will lead us to the concept of \textit{topological vertex coalgebra}.

Let us recall that if $U$ is any vector space, we can endow its dual space $V=U^*$ with the linearly compact topology, where a basis of open neighborhoods of zero in $V$ is given by the annihilators $W^{\perp}$ of the finite-dimensional subspaces $W$ of $U$. We will call vector spaces equipped with this topology \textit{linearly compact spaces}. Note that for any finite-dimensional vector space $V$ this is the discrete topology.  

This gives us a contravariant functor $^*$ from the category $\mathcal{V}ec$ of vector spaces to the category $\mathcal{LCV}ec$ of linearly compact spaces (where the morphisms are the continuous linear maps). This functor has an inverse, the continuous dual, that takes each linearly compact space $V$ to its space $V'$ of all continuous linear functionals $f: V \rightarrow \C$, and therefore the functor $^*$ is an antiequivalence of categories.

If $V_1$ and $V_2$ are two linearly compact spaces, we can define their \textit{completed tensor product}  $V_1 \hat{\otimes} V_2$ as the completion of $V_1 \otimes V_2$ with respect to the tensor product topology, where a basis of open neighborhoods of zero is given by the spaces $W_1 \otimes V_2 \, + \, V_1 \otimes W_2$ with $W_i$ open neighborhood of zero in $V_i$ for $i=1,2$. This allows us to consider the monoidal structure on the category $\mathcal{LCV}ec$ where the monoidal product is $\hotimes$ and the monoidal unit is $\C$. It is well-known that $(U_1 \otimes U_2)^* \simeq U_1^* \hat{\otimes} U_2^*$ for all vector spaces $U_1, U_2$ (c.f. \cite{Dieudonne}), and therefore the dualization functor * preserves the monoidal structure.

Now let us take $U$ to be a vertex algebra instead of an arbitrary vector space. Since for each $n \in \Z$ we have the $n$-th product $\cdot _{(n)} \cdot : U \otimes U \rightarrow U$, its dual map (which is continuous) will be called \textit{$n$-th coproduct}, and we shall denote it by
\begin{equation}\label{nth_products_def_U_dual}
\Delta_{n}: V \rightarrow V \hat{\otimes} V.
\end{equation}
We can collect all of these maps in a generating series:
\begin{equation}\label{Ydual_def_U_dual}
\Ydual (x)=\sum_{n \in \Z} \Delta_{n}x^{-n-1}.
\end{equation}
Due to the fact that $V \hotimes V = (U \otimes U)^*$, we can write for $v \in V$ and $u, u' \in U$ $$(\Ydual(x)v)(u \otimes u') = v(Y(u,x)u').$$
Similarly, the vacuum vector $\1$ induces a continuous \textit{covacuum map}:
\begin{align}
c: V &\rightarrow \C \label{covacuum_def_U_dual}\\
f &\mapsto f(\1). \nonumber
\end{align}

We have arrived to a structure very similar to the one that Hubbard has called \textit{vertex coalgebra} in \cite{Hubbard_vertex_coalgebras}. Nonetheless, there are some important differences, mostly due to the topological aspects that we must consider, which do not appear in Hubbard's work because he studied graded dual spaces.

\begin{definition}
A \textit{topological vertex coalgebra} consists of a linearly compact space $V$ together with a continuous linear map
\begin{align}
\Ydual (x): V &\rightarrow (V \hat{\otimes} V)[[x,x^{-1}]]\label{Y_dual_def}\\
v &\mapsto \Ydual (x)v=\sum_{n \in \Z} \Delta_{n}(v)x^{-n-1}. \nonumber
\end{align}
and a continuous linear map $c: V \rightarrow \C$ called \textit{covacuum map}, satisfying the following axioms:
\begin{enumerate}
\item Convergence: 
\begin{equation}\label{punctual_convergence}
\lim_{n\to\infty}\Delta_{n}(v) = 0 \quad \text{uniformly for }v\in V.
\end{equation} 
\item Left counit: For all $v \in V$,
 \begin{equation}\label{Left_counit}
 (c \hotimes Id)\,\Ydual (x)v = v.
 \end{equation}
\item Cocreation: For all $v \in V$, $(Id \hotimes c) \Ydual (x)v$ is a holomorphic formal power series in $x$ with coefficients in $V$ and
 \begin{equation}
 (Id \hotimes c) \Ydual (x)v \big|_{x=0}= v.\label{Cocreation}
 \end{equation}
\item Co-Jacobi identity:
\begin{align}
x_0^{-1}&\delta\left(\frac{x_1-x_2}{x_0}\right)(Id \hotimes \Ydual(x_2)) \Ydual (x_1) - x_0^{-1}\delta\left(\frac{x_2-x_1}{-x_0}\right)(\tau \hotimes Id)  \label{CoJacobi} \\
&(Id \hotimes \Ydual(x_1)) \Ydual (x_2)=x_2^{-1}\delta\left(\frac{x_1-x_0}{x_2}\right)(\Ydual(x_0) \hotimes Id)\Ydual(x_2),\nonumber
\end{align}
where $\tau: V \hotimes V \rightarrow V \hotimes V$ is the flip operator $a \hotimes b \mapsto b \hotimes a$.
\end{enumerate}
\end{definition}

\begin{remark}\label{good_definition_TVC}
There are several aspects of this definition that require some clarification. Firstly, since the product topology is the topology of punctual convergence, the convergence axiom is equivalent to requiring that the singular part of the formal series $\Ydual(x)(v)$, that is
\begin{equation}\label{convergence_series}
\sum_{n=0}^{\infty} \Delta_{n}(v)\,x^{-n-1} \in (V \hotimes V)[[x^{-1}]],
\end{equation}
be uniformly convergent for all $v \in V$ when we set the product topology over $(V \hotimes V)[[x^{-1}]]=\prod_{n \geq 0}(V \hotimes V)x^{-n-1}$, where each factor is homeomorphic to $V \hotimes V$. This may be interpreted as a sort of ``generalized truncation'' axiom, changing the truncation of the series (\ref{convergence_series}) for its convergence. In a similar fashion, the continuity of $\Ydual(x)$ is equivalent to asking for all the $n$-th coproducts to be continuous maps.

On the other hand, the axiom of convergence plays the same role for topological vertex coalgebras that the truncation axiom plays for vertex algebras: it allows us to multiply the formal power series that appear in the co-Jacobi identity. Indeed, for any $l,t,j \in \Z$, comparing of the coefficients of $x_{0}^{-l-1}x_{1}^{-t-1}x_{2}^{-j-1}$ in each side of the co-Jacobi identity gives the formula
\begin{align}
\sum_{i=0}^{\infty}&(-1)^{i}\binom{l}{i}(Id \hotimes \Delta_{j+i})\Delta_{t+l-i}-(\tau\, \hat{\otimes}\, Id)\sum_{i=0}^{\infty}(-1)^{l+i}\binom{l}{i}(Id \hotimes \Delta_{t+i})\Delta_{j+l-i}\nonumber\\
&=\sum_{i=0}^{\infty}\binom{t}{i}(\Delta_{l+i} \hotimes Id)\Delta_{t+n-i}.\label{coJacobi_coefficients}
\end{align}
We need to check that fixing $v \in V$, all of these sums applied to $v$ converge in $V \hotimes V \hotimes V$. Let us take some open neighborhood of zero $W$ in this space; by definition, we may assume that is has the form $W=W_1 \hotimes V \hotimes V + V \hotimes W_2$ with $W_1$ and $W_2$ open in $V$ and $V \hotimes V$ respectively. Now (\ref{punctual_convergence}) allows us to choose some $N \geq 0$ such that $\Delta_{i}(w) \in W_2$ for all $i \geq N$ and for all $w \in V$. But then $$\sum_{i=\text{max}(0,N-j)}^{\infty}(-1)^{i}\binom{l}{i}(Id \hotimes \Delta_{j+i})\Delta_{t+l-i}(v) \in V \hotimes W_2 \subseteq W,$$so the first term of (\ref{coJacobi_coefficients}) is a convergent series, and we can proceed\linebreak analogously with the other two terms. As a last comment, note that this proof relies heavily on the fact that the convergence (\ref{punctual_convergence}) is uniform.
\end{remark}

In order to define the category of topological vertex coalgebras, which we shall denote by $\mathcal{TVC}$, we need the following definition.

\begin{definition}
A \textit{homomorphism} of topological vertex coalgebras is a continuous linear map $f: V \rightarrow \tilde{V}$ between topological vertex coalgebras such that $$c_{\tilde{V}}  f = c_{V} \quad \text{and} \quad \Ydual_{\tilde{V}}(x) f=(f \hotimes f) \Ydual_{V}(x),$$ with $f \hotimes f$ extended coefficient-wise.
\end{definition}

We are now ready to prove the duality between vertex algebras and topological vertex coalgebras.

\begin{theorem}\label{first_vertex_Cartier}
If $U$ is a vertex algebra, then its dual space $V=U^*$ is a topological vertex coalgebra with the linearly compact topology and the maps $\Ydual(x)$ and $c$ defined by (\ref{Ydual_def_U_dual}) and (\ref{covacuum_def_U_dual}).

Conversely, if $V$ is topological vertex coalgebra, then its continuous dual space $U=V'$ is a vertex algebra.

Moreover, these correspondences define an antiequivalence of categories between $\mathcal{VA}$ and $\mathcal{TVC}$.
\end{theorem}

\begin{proof}
Let $U$ be a vertex algebra. We already know that $V=U^*$ is a linearly compact space, and we also have that for each $n \in \Z$, $\Delta_n$ is continuous, which implies that $\Ydual(x)$ is continuous when we consider the space $(V \hotimes V)[[x,x^{-1}]]$ as the product $\prod_{n \in \Z}(V \hotimes V)x^{-n-1}$ with the product topology, where each factor is homeomorphic to $V \hotimes V$.

In order to prove that $V$ is a topological vertex coalgebra, it remains to check the convergence axiom as well as conditions (\ref{Left_counit}) to (\ref{CoJacobi}). The latter follow quite directly from their respective axioms in the definition of vertex algebra, as Hubbard shows in \cite{Hubbard_vertex_operator_coalg}. For example, the left unit and counit axioms follow from one another because if $v \in V$ and $u \in U$, then
\begin{align*}
((c \hotimes Id)\Ydual(x)v)(u) &= (\Ydual(x)v) (\1 \otimes u)\\
&=v(Y(\1,x)u),
\end{align*}
and therefore if either $(c \hotimes Id)\Ydual(x)$ or $Y(\1,x)$ is the identity, so is the other.

Therefore we only need to prove the convergence axiom (\ref{punctual_convergence}). We begin by rephrasing the truncation axiom: it is equivalent to asking for the increasing sequence of spaces $$A_{N}:=\langle \{ u \otimes v \in U \otimes U: u_{(n)}v=0 \text{ for all }n \geq N \} \rangle$$to be a filtration of $U \otimes U$. Let us choose any increasing filtration $\{B_{N}\}_{N \geq 0}$ of $U \otimes U$ by finite-dimensional spaces, and define $C_{N}=A_{N} \cap B_{N}$ for all $N \geq 0$, which is still an increasing filtration of $U \otimes U$. Now its sequence of annihilators $\{W_{N}\}_{N \geq 0}$ is a basis of open neighborhoods of zero in $V \hotimes V$, with the property that for each $N \geq 0$, $$\Delta_{n}(f) \in W_{N} \quad \forall f \in V, \forall n \geq N,$$ which is exactly what (\ref{punctual_convergence}) means.

Reciprocally, if $V$ is a topological vertex algebra and we define $U=V'$, we know that the continuous dual functor sends the maps $\Delta_{n}$ for $n \in \Z$ and $c$ to some linear maps 
\[
\cdot_{(n)}\cdot : U \otimes U \rightarrow U \quad \text{and} \quad \eta: \C \rightarrow U.
\]
Let $\1=\eta(1)$ and $Y(u,x)v=\sum_{n \in \Z}u_{(n)}v\,x^{-n-1}$ for all $u,v \in U$. Now once again, all the axioms of a vertex algebra follow directly from their\linebreak corresponding ones of topological vertex coalgebras, except maybe the truncation condition.

Let $u,v \in U$. Then $u \otimes v \in U \otimes U \simeq (V \hotimes V)'$, so Ker$(u \otimes v)$ is a subspace of $V \hotimes V$ of codimension at most 1, and is therefore an open neighborhood of zero. Now we can use (\ref{punctual_convergence}): there exists $N \geq 0$ such that $\Delta_{m}(f) \in \text{Ker}(u \otimes v)$ for all $f \in V$ and all $n \geq N$, which is easily seen to be equivalent to the truncation axiom for $U$.

Finally, if $f: U \rightarrow \tilde{U}$ is a linear map between vertex algebras, $V=U^*$ and $\tilde{V}=\tilde{U}^*$, then $f^*:\tilde{V} \rightarrow V$ is continuous linear map and for all $u,u'\in U$ and $v \in \tilde{V}$,
\begin{align*}
(\Ydual_{V}(x)(f^*v))(u \otimes u')&= (f^*v) (Y_{U}(u,x)u')\\
&=v(f(Y_{U}(u,x)u'))
\end{align*}
and
\begin{align*}
((f^* \hotimes f^*)\Ydual_{\tilde{V}}(x)v)(u \otimes u')&=(\Ydual_{\tilde{V}}(x)v)(f(u) \otimes f(u'))\\
&=v(Y_{\tilde{U}}(f(u),x)f(u')).
\end{align*}
Similarly, $f(\1) = \1 \otimes \1$ if and only if $c_{\tilde{V}}\,f^*=c_{V}$, so $f$ is a morphism in $\mathcal{VA}$ if and only if $f^*$ is a morphism in $\mathcal{TVC}$. Therefore the dual and continuous dual functors are well-defined when restricted to these smaller categories, establishing the antiequivalence we wanted.
\end{proof}

We will now use this result to endow $\mathcal{TVC}$ with a monoidal structure.

\begin{proposition}\label{TVC_monoidal}
Let $V_1$ and $V_2$ be topological vertex coalgebras. Then the linearly compact space $V_1 \hotimes V_2$ can be given a topological vertex coalgebra structure by taking for all $v_1 \in V_1$ and $v_2 \in V_2$
\begin{align}
\Ydual_{V_1 \hotimes V_2}(x)(v_1 \hotimes v_2)&=(Id \hotimes \tau \hotimes Id)\, (\Ydual_{V_1}(x)(v_1) \hotimes \Ydual_{V_2}(x)(v_2)),\label{Ydual_tensor_prod}\\
c_{V_1 \hotimes V_{2}}(v_1 \hotimes v_2)&= c_{V_{1}}(v_1)\, c_{V_2}(v_2).\label{covacuum_tensor_prod}
\end{align}
In particular, $\mathcal{TVC}$ is a monoidal subcategory of $\mathcal{LCV}ec$.
\end{proposition}

\begin{proof}
Let $U_i=V_i'$ for $i=1,2$. Since each $U_i$ is a vertex algebra, we can consider their tensor product $U_1 \otimes U_2$, which is again a vertex algebra. Now its dual space is isomorphic to $V_1 \hotimes V_2$, and therefore by the preceding theorem it is a topological vertex coalgebra. Thus we only need to check that its structure maps are given by (\ref{Ydual_tensor_prod}) and (\ref{covacuum_tensor_prod}). For all $v_i \in V_i$ and $u_i, u_i' \in U_i$ ($i=1,2$), we have

\begin{align*}
\Ydual_{V_1 \hotimes V_2}&(x)(v_1 \hotimes v_2)(u_1 \otimes u_2 \otimes u_1' \otimes u_2')=\\
&=(v_1 \hotimes v_2) (Y_{U_1 \otimes U_2}(u_1 \otimes u_2,x)(u_1' \otimes u_2'))\\
&=(v_1 \hotimes v_2)(Y_{U_1}(u_1,x) u_1' \otimes Y_{U_2}(u_2,x)u_2')\\
&=v_1(Y_{U_1}(u_1,x) u_1) v_2(Y_{U_2}(u_2,x) u_2')\\
&=\Ydual_{V_1}(x)(v_1)(u_1 \otimes u_1') \Ydual_{V_2}(x)(v_2)(u_2 \otimes u_2')\\
&=(\Ydual_{V_1}(x)(v_1) \hotimes \Ydual_{V_2}(x)(v_2) )(u_1 \otimes u_1' \otimes u_2 \otimes u_2')\\
&=((Id \hotimes \tau \hotimes Id)(\Ydual_{V_1}(x)(v_1) \hotimes \Ydual_{V_2}(x)(v_2)))(u_1 \otimes u_2 \otimes u_1' \otimes u_2').
\end{align*}

The proof of (\ref{covacuum_tensor_prod}) is analogous.
\end{proof}

We shall close this section with a brief discussion on graded vertex algebras and coalgebras. A vertex algebra $U$ is called \textit{graded} if it has a $\Z$-grading $U=\bigoplus_{t \in \Z}U_{(t)}$ such that $U_{(t)}=0$ for $t << 0$ and for all $u \in U_{(t)}$, $v \in U_{(s)}$ and $n \in \Z$, we have $u_{(n)}v \in U_{(t+s-n-1)}$. Let us consider the graded dual space $U^{o} = \bigoplus_{t \in \Z}U_{(t)}^*$. Let $c^{o}: U^{o} \rightarrow \C$ be the double dual of $\1$ and define for any $n \in \Z$ the map $\Delta_{n}^{o}: U^{o} \rightarrow U^{o} \otimes U^{o}$ by $\Delta_{n}^{o}(v)(u \otimes u')=v(u_{(n)}u')$ for all $v \in U^{o}$ and $u,u' \in U$. It was proven in \cite{Hubbard_vertex_coalgebras} that $U^{o}$ with these maps is a \textit{graded vertex coalgebra}, which is defined by a set of axioms similar to those of topological vertex coalgebras but exempt of the topological aspects. In particular, the axiom of convergence is replaced by the following grading condition: for all $v \in U^{o}_{(t)}$,
\begin{equation}\label{grading_condition_Hubbard}
\Delta_{n}^{o}(v) \in (U^{o} \otimes U^{o})_{(t + n + 1)}.  
\end{equation}
Here the grading in $U^{o} \otimes U^{o}$ is given by declaring $\text{deg}(v_{1} \otimes v_{2})=k$ whenever $v_{1} \in U_{(k_{1})}^{o}$, $v_{2} \in U_{(k_{2})}^{o}$ and $k_1 + k_2 = k$.

On the other hand, if we take the full dual $U^*$ and the continuous dual maps $\Delta_{n}: V \rightarrow V \hotimes V$ and $c: V \rightarrow \C$ defined by (\ref{nth_products_def_U_dual}) and (\ref{covacuum_def_U_dual}), we obtain a topological vertex coalgebra by Theorem \ref{first_vertex_Cartier}. Now $U^*$ is the completion of the discrete space $U^{o}$, and the maps $\Delta_{n}$ and $c$ are continuous extensions of $\Delta_{n}^{o}$ and $c^{o}$ respectively. It is illustrating to see how the convergence axiom for $U^*$ can be deduced in this case directly from the grading condition (\ref{grading_condition_Hubbard}) for $U^{o}$. Let $T\in \Z$ such that $U_{t}=0$ for all $t < T$, and define a decreasing filtration by open subspaces $\{W_{k}\}_{k \geq T}$ of $V \hotimes V$ as $W_{k}= \prod_{t \geq k} (U^{o} \otimes U^{o})_{(t)}$. Now given any $k \in \Z$, we may take $N=k-T-1$, and then the grading condition implies that $\Delta_{n}(V) \subseteq W_k$ for all $n \geq N$, proving thus the axiom of convergence.

\section{Vertex bicoalgebras and Cartier duality}\label{subsection_vertex_coalg}

We will now continue the program of studying the dual spaces of vertex algebras by examining the case of vertex bialgebras. Our goal is to prove a version of Cartier duality for vertex algebras.

We begin by recalling the classical Cartier duality. If we restrict our attention to coalgebras instead of arbitrary vector spaces, we find that for any given (coassociative counital) coalgebra $(U,\Delta, \epsilon)$ the dualization functor sends $$\Delta: U \rightarrow U \otimes U \quad \text{and} \quad \epsilon: U \rightarrow \C$$ to the continuous linear maps 
\begin{equation*}
\Delta^*: U^* \hotimes U^* \rightarrow U^*  \quad \text{and} \quad \epsilon^*: \C \rightarrow U^*,
\end{equation*}
and in fact $(U^*,\Delta^*,\epsilon^*)$ is an (associative unital) algebra in the monoidal\linebreak category $\mathcal{LCV}ec$, namely, a \textit{linearly compact algebra}. Conversely, the con-\linebreak tinuous dual functor sends any linearly compact algebra to a coalgebra. Thus the dualization functor restricted to the categories $\mathcal{C}lg$ of coalgebras and $\mathcal{LCA}lg$ of linearly compact algebras is once again an antiequivalence of categories, and this is precisely the statement known as the first version of Cartier duality \cite{Hazewinkel_formalgroups}: $$\mathcal{C}lg \simeq \mathcal{LCA}lg^{op}.$$

Naturally, in the last paragraph we could have considered algebras in $\mathcal{V}ec$ instead of coalgebras, and we would have arrived at linearly compact coalgebras instead of algebras. 
Combining both reasonings, we obtain the full version of Cartier duality, relating the categories of bialgebras in $\mathcal{V}ec$ and linearly compact bialgebras: $$\mathcal{B}lg \simeq \mathcal{LCB}lg^{op}.$$

There is one further step that we could take that consists in replacing bialgebras by Hopf algebras, but we will not make use of it in this work, mainly because the notion of Hopf vertex algebra is yet to be defined.

We now return to the realm of vertex algebras. Since a vertex bialgebra $U$ is in particular a coalgebra, we know that its dual $V=U^*$ is a linearly compact algebra. Moreover, $V$ is a commutative algebra whenever $U$ is a cocommutative coalgebra (e.g., when $U=\U(R)$).

This leads us to the following dualization of the notion of vertex bialgebra.

\begin{definition}
A \textit{vertex bicoalgebra} is a topological vertex coalgebra $V$ together with a structure of algebra in $\mathcal{LCV}ec$ such that its algebra maps $$\cdot: V \hotimes V \rightarrow V \quad \text{and} \quad \eta: \C \rightarrow V$$ are topological vertex coalgebra homomorphisms. That is, a vertex bicoalgebra is simply an algebra in the monoidal category $\mathcal{TVC}$. 
\end{definition}

Let us recall that bialgebras may equivalently be defined as being coalgebras over the category of algebras or algebras over the category of coalgebras, since these structures are dual to each other. For vertex bialgebras such a property is not to be expected since it is not clear how $U((x))$ can be given a coalgebra structure such that the map $Y$ becomes a coalgebra map.

Nonetheless, there is a way to obtain such a result for vertex bicoalgebras. The first step is to note that if $V$ is a vertex bicoalgebra, the image of the map $\Ydual(x)$ lies in the subspace of $(V \hotimes V) [[x,x^{-1}]]$ consisting of series of the form $\sum_{n \in \Z}a_n x^{-n-1}$ whose coefficients satisfy that $a_n \rightarrow 0$ in $V \hotimes V$ when $n \rightarrow \infty$. We denote this space as $(V \hotimes V) [[x,x^{-1}]]^{conv}$, since it is formed by formal series whose singular part is point-wise convergent.

\begin{lemma}
If $W$ is a linearly compact algebra, then $W[[x,x^{-1}]]^{conv}$ is an algebra, containing $W((x))$ as a subalgebra. If $W$ is finite-dimensional, both algebras coincide.
\end{lemma}

\begin{proof}
Let $a(x)=\sum_{n\in \Z}a_n x^{-n-1}$ and $b(x)=\sum_{n \in \Z}b_n x^{-n-1}$ be elements of $W[[x,x^{-1}]]^{conv}$. The product induced from $W$ is
$c(z)=\sum_{n\in \Z} c_n x^{-n-1}$, where $c_n=\sum_{m\in \Z} a_{m} \cdot b_{n-m-1}$, and we only need to verify that it is well-defined. That is, $c_n$ must be a well-defined element of $W$ for any $n \in \Z$ and also $c_n \rightarrow 0$ for $n \rightarrow \infty$.

Fixing $n \in \Z$, we have that $c_n$ exists if and only if both $\sum_{m=0}^{\infty}a_{m} \cdot b_{n-m-1}$ and $\sum_{m=0}^{\infty}a_{-m-1} \cdot b_{n+m}$ converge in $W$, which in turns happens if and only if both of their sequences of partial sums are Cauchy sequences. We will only prove that $\{\sum_{m=0}^{M}a_{m} \cdot b_{n-m-1}\}_{M \geq 0}$ is a Cauchy sequence, the other one being analogous. Let us choose a basis of open neighborhoods $\{U_k\}_{k \geq 0}$ of $W$ where each $U_k$ is also a two-sided ideal (for example, extending the proof of \cite[Ch.~1,~Prop.~1]{Dieudonne} to the non-commutative case). The convergence of $a(x)$ and $b(x)$ implies that for any $k \geq 0$ there exist $M_1,M_2 \geq 0$ such that $a_m \in U_k$ for $m \geq M_1$ and $b_m \in U_k$ for $m \geq M_2$. Taking $M=\text{max}\{M_1,M_2\}$, we have that $\sum_{m=M'}^{M''}a_{m} \cdot b_{n-m-1} \in U_k \cdot W \subseteq U_k$ for any $M''\geq M' \geq M$. Thus it is a Cauchy sequence, which proves that each $c_n$ is well defined.

Moreover, if $N=M_1+M_2$ and $n > N$, then for all $m \geq 0$ we have that either $m \geq M_1$ or $n-m-1 \geq M_2$, so $\sum_{m=0}^{\infty}a_{m} \cdot b_{n-m-1} \in U_k$. Performing similar computations with $\sum_{m=0}^{\infty}a_{-m-1} \cdot b_{n+m}$ allows us to conclude that $c_n \rightarrow 0$ for $n \rightarrow \infty$. Finally, $W((x))$ is clearly a subalgebra since the coefficients of the singular part of its elements become null from a certain point onwards, and if $\text{dim }W < \infty$ these are the only sequences that converge to $0$ because $W$ is discrete in this case.
\end{proof}

We may now state and prove the desired result.

\begin{proposition}
Let $(V,\Ydual(x),c)$ be a topological vertex coalgebra that also has a linearly compact algebra structure $(V,\cdot,\eta)$. Then $V$ is a vertex bicoalgebra if and only if the maps $$\Ydual(x): V \rightarrow (V \hotimes V) [[x,x^{-1}]]^{conv} \quad \text{and} \quad c: V \rightarrow \C$$are algebra homomorphisms.
\end{proposition}

\begin{proof}
For $V$ to be a vertex bicoalgebra, the following statements must hold true

\begin{enumerate}
\item $\Ydual_{V}(x)(v \cdot w) = (\cdot \hotimes \cdot) \Ydual_{V \hotimes V}(x)(v \hotimes w)$ for all $v,w \in V$.
\item $c_{V}(v \cdot w) = c_{V \hotimes V}(v \hotimes w)$ for all $v,w \in V$.
\item $\Ydual_{V}(x) \eta= (\eta \hotimes \eta) \Ydual_{\C}(x)$.
\item $c_{V} \eta = c_{\C}$.
\end{enumerate}

On the other hand, the maps $\Ydual(x)$ and $c$ are algebra homomorphisms if and only if we have

\begin{enumerate}[label=(\arabic*')]
\item $\Ydual(x)(v \cdot w) = \Ydual(x)(v) \cdot \Ydual(x)(w)$ for all $v,w \in V$.
\item $\Ydual(x)(1)= 1 \hotimes 1$.
\item $c(v \cdot w) = c(v)\,c(w)$ for all $v,w \in V$.
\item $c(1) = 1$.
\end{enumerate}

It is clear that (2) and (3') are equivalent, as well as (3) with (2') and (4) with (4'). Thus we only need to establish the equivalence (1) $\Longleftrightarrow$ (1'). But this holds because for any $v,w \in V$, we have
\begin{align*}
(\cdot \hotimes \cdot) \Ydual_{V \hotimes V}(x)(v \hotimes w) &= (\cdot \hotimes \cdot) (Id \hotimes \tau \hotimes Id)\, (\Ydual(x)(v) \hotimes \Ydual(x)(w))\\
&=\cdot_{V \hotimes V}\, (\Ydual(x)(v) \hotimes \Ydual(x)(w))\\
&=\cdot_{V \hotimes V}\,\sum_{n \in \Z} \sum_{m \in \Z}\Delta_{m}(v) \hotimes \Delta_{n-m-1}(w)\,x^{-n-1}\\
&=\sum_{n \in \Z} \sum_{m \in \Z}\Delta_{m}(v) \cdot\Delta_{n-m-1}(w)\,x^{-n-1}\\
&=\Ydual(x)(v) \cdot \Ydual(x)(w).
\end{align*}
\end{proof}
Let us denote the category of vertex bicoalgebras by $\mathcal{VBC}$, where the morphisms are those morphisms in $\mathcal{TVC}$ that preserve the algebra structure. With these definitions, the Cartier duality for vertex bialgebras may be stated as follows.

\begin{proposition}[Cartier duality]\label{Cartier_vertex_bialgebras}
If $U$ is a vertex bialgebra, then its dual space $V=U^*$ is a vertex bicoalgebra. Conversely, if $V$ is a vertex bicoalgebra, then its continuous dual space $U=V'$ is a vertex bialgebra.

Moreover, these correspondences define an antiequivalence of categories between $\mathcal{VBA}$ and $\mathcal{VBC}$.
\end{proposition}

\begin{proof}
It follows immediately from putting together Theorem \ref{first_vertex_Cartier}, the first version of the classical Cartier duality and the fact that the structure of topological vertex coalgebra defined in $V \hotimes V$ is exactly the one that turns the isomorphism $V \hotimes V \simeq (U \otimes U)^*$ of linearly compact spaces into an isomorphism of topological vertex coalgebras.
\end{proof}

\section{Formal vertex laws}\label{section_formal_vertex_laws}

We will now describe in detail the construction of the formal group law associated to an infinite-dimensional Lie algebra $\mathfrak{g}$, in order to understand how we may apply that construction to the conformal case.

Let $\mathfrak{g}$ be an infinite-dimensional Lie algebra with an ordered basis\linebreak $\{e_i : i \in \mathcal{I}\}$, where $\mathcal{I}$ is some totally ordered infinite set. Let $|\mathcal{I}|$ be its cardinal, and let $U=\U(\mathfrak{g})$. One way of defining the PBW basis for $U$ is as follows: for any $\textbf{k}: \mathcal{I} \rightarrow \Z_{+}$, $i \mapsto k_i$ with finite support, we denote
\begin{equation}\label{elemento_base_U}
e_{\textbf{k}} \, = \, e_{i_1}^{k_1} \cdots e_{i_n}^{k_n},
\end{equation}
where $\{i_1,...,i_n\}$ is the support of $\textbf{k}$, labelled in such a way as to form an increasing sequence in $\mathcal{I}$. The set of all such maps $\textbf{k}$ forms a monoid, which will be denoted as $\textbf{K} = \Z_{+}^{(\mathcal{I})}$. The \textit{norm} of $\textbf{k} \in \textbf{K}$ is $|\textbf{k}|=\sum_{i \in \mathcal{I}}k_{i}$. Now the PBW basis for $U$ is the set $\{e_{\textbf{k}} : \textbf{k} \in \textbf{K} \}$.

We shall denote its \textit{dual pseudobasis} by $\{ e^{\textbf{k}} : \textbf{k} \in \textbf{K} \}$, where each $e^{\textbf{k}} \in U^*$ is determined by $e^{\textbf{k}}(e_{\textbf{k}'})=\delta_{\textbf{k}\textbf{k}'}$ for any $\textbf{k}'\in\textbf{K}$ and $\delta$ denotes the Kronecker delta. It is not a basis in the classical sense because a generic element $f\in U^*$ can be written as $$\dst f = \sum_{\textbf{k} \in \textbf{K}}f(e_{\textbf{k}})e^{\textbf{k}}$$ but this sum is not necessarily finite, so we have $U^*=\prod_{\textbf{k} \in \textbf{K}} \C e^{\textbf{k}}$.

Let $A$ be the algebra $\C
[[X_{i} : i\in \mathcal{I}]]$. We shall denote monomials in $A$ as $X^{\textbf{k}}=\prod_{i \in \mathcal{I}} X_{i}^{k_{i}},$ where $\textbf{k} \in \textbf{K}$, so that $X^{\textbf{k}+\textbf{k}'}=X^{\textbf{k}}X^{\textbf{k}'}$ for all $\textbf{k}, \textbf{k}' \in \textbf{K}$. Setting $\textbf{k}!=\prod_{i \in \mathcal{I}}k_{i}! $ and $X^{(\textbf{k})}=\frac{1}{\textbf{k}!}X^{\textbf{k}}$, we may state the following classical result, whose proof may be found in many textbooks such as \cite{Serre, Hazewinkel_formalgroups}.

\begin{proposition}\label{iso_U_dual}
The map $\psi: A \rightarrow U^*$ given by $X^{(\textbf{k})} \mapsto e^{\textbf{k}}$ is an algebra isomorphism.
\end{proposition}


Now this isomorphism allows us to translate the coalgebra structure from $U^*$ to $A$. Let us define for each $l \in \mathcal{I}$ the formal power series $$F_l(X,Y) = \Delta (X_l) \in A \hotimes A = \C[[ X_i, Y_i : i \in \mathcal{I}]],$$where $X_{i}$ and $Y_{i}$ stand for $X_{i} \hotimes 1$ and $1 \hotimes X_{i}$ in $A \hotimes A$, $X=\{X_i\}_{i \in \mathcal{I}}$ and $Y=\{Y_i\}_{i \in \mathcal{I}}$. Let $F(X,Y)=\{F_l(X,Y)\}_{l \in \mathcal{I}} \in \C[[X_{i},Y_{i}: i \in \mathcal{I}]]^{\mathcal{I}}$. One can easily show that the counit and coassociativity axioms on $V$ are equivalent to the following properties for $F(X,Y)$:
\begin{align}
&F(X,0)=X=F(0,X)\label{unit_formal_group_law}\\
&F(X,F(Y,Z))=F(F(X,Y),Z)\label{associativity_formal_group_law}
\end{align}

On the other hand, if we write $$F_{l}(X,Y)=\sum_{\textbf{k},\textbf{k}' \in \textbf{K}} c_{\textbf{k},\textbf{k}'}^{l} X^{\textbf{k}} Y^{\textbf{k}'}$$with $c_{\textbf{k},\textbf{k}'}^{l} \in \C$ for all $l \in \mathcal{I}$ and $\textbf{k},\textbf{k}' \in \textbf{K}$, we have that these coefficients satisfy the following finiteness condition: for any fixed pair $\textbf{k},\textbf{k}' \in \textbf{K}$, there exist finitely many $l \in \mathcal{I}$ such that $c_{\textbf{k},\textbf{k}'}^{l} \neq 0$.

A set of formal power series $F(X,Y)$ satisfying the equations above and the finiteness condition is called an \textit{infinite-dimensional formal group law}. That condition allows us to write the associativity property, since we can use it to show that the coefficients in (\ref{associativity_formal_group_law}) are well-defined \cite{Hazewinkel_formalgroups}.

From now on, we begin our construction of the vertex analogue of the concept of formal group law. First we notice that if $R$ is a Lie conformal algebra, we can take $\mathfrak{g}=R_{Lie}$ in the above construction, and some of the results still hold. For example, since $U(R)$ and $\U(R_{Lie})$ have the same coalgebra structure, their duals are isomorphic as algebras, and thus Proposition \ref{iso_U_dual} is still valid for $U=\U(R)$. This means that using the same argument as above, we can translate the topological vertex coalgebra structure of $\U(R)^*$ onto $A$, and encode this extra structure through certain formal power series that shall play the role of formal group laws in this setting. While viewing the algebra $A$ as a vertex bicoalgebra, we shall denote it as $\textbf{V}$.

Let us define for each $n \in \Z$ and for any $l \in \mathcal{I}$ the formal power series 
\begin{equation}\label{def_Fijm}
F_{l}^n(X,Y) = \Delta_n (X_{l}) \in \textbf{V} \hat{\otimes} \textbf{V} = \C[[X_{i},Y_{i} : i \in \mathcal{I}]],
\end{equation} 
and let us also set the following notation:
\begin{align}
&F_{l}(x)(X,Y)=\sum_{n \in \Z} F_{l}^n(X,Y) x^{-n-1} \in (\textbf{V} \hotimes \textbf{V})[[x,x^{-1}]],\label{def_Fijx}\\
&F(x)(X,Y)=\{F_{l}(x)(X,Y) \}_{l \in \mathcal{I}} \in (\textbf{V} \hotimes \textbf{V})[[x,x^{-1}]]^{\mathcal{I}}.\label{def_Fx}
\end{align}

We will now proceed to rewrite the axioms of a topological vertex coalgebra in $\textbf{V}$ in terms of properties of $F(x)(X,Y)$.

\begin{proposition}\label{def_formal_vertex_law}
The set of formal power series $F(x)(X,Y)$ defined by (\ref{def_Fx}) satisfies
\begin{enumerate}
\item Convergence: For any $r \in \Z_+$ and any finite $I \subseteq \mathcal{I}$, let $\mathcal{M}_{r,I}$ be the closure of the ideal $$ \left\langle \left\{ X^{\textbf{k}}Y^{\textbf{k}'} \; \Big| \; r < |\textbf{k}|+|\textbf{k}'| \wedge \exists j \notin I : k_{j} \neq 0 \vee k_{j}' \neq 0 \right\} \right\rangle.$$ 
Then there exists $N_{r,I} \in \Z_+$ such that for all $l \in \mathcal{I}$ and $n \geq N_{r,I}$,
\begin{equation}\label{convergence_vertex_formal_law}
F_{l}^n(X,Y) \equiv 0 \quad \textup{mod $\mathcal{M}_{r,I}$}
\end{equation}
\item Left identity:
\begin{equation}\label{left_identity_CFGL}
F(x)(0,X)=X
\end{equation}
\item Right identity: Each formal power series $F_{i}(x)(X,0)$ is holomorphic in $x$ and
\begin{equation}
F(x)(X,0)\big|_{x=0}=X.\label{right_identity_CFGL}
\end{equation}
\item Jacobi identity:
\begin{align}
x_0^{-1}\delta&\left(\frac{x_1-x_2}{x_0}\right)F(x_1)(X,F(x_2)(Y,Z)) \nonumber\\
-&x_0^{-1}\delta\left(\frac{x_2-x_1}{-x_0}\right)F(x_2)(Y,F(x_1)(X,Z)) \label{Jacobi_CFGL}\\
&=x_2^{-1}\delta\left(\frac{x_1-x_0}{x_2}\right)F(x_2)(F(x_0)(X,Y),Z).\nonumber
\end{align}
\end{enumerate}
\end{proposition}

\begin{proof}
Since $\{X^{\textbf{k}} \hotimes X^{\textbf{k}'} : \textbf{k}, \textbf{k}' \in \textbf{K} \}$ is a pseudobasis for $\textbf{V} \hotimes \textbf{V}$, for any given $l \in \mathcal{I}$ and $n \in \Z$ we can write 
\[
F_{l}^n(X,Y) = \sum_{\textbf{k},\textbf{k}' \in \textbf{K}} c_{\textbf{k}\textbf{k}'}^{ln}\, X^{\textbf{k}} \hotimes X^{\textbf{k}'} = \sum_{\textbf{k},\textbf{k}' \in \textbf{K}} c_{\textbf{k}\textbf{k}'}^{ln}\, X^{\textbf{k}} Y^{\textbf{k}'}
\]
where $c_{\textbf{k}\textbf{k}'}^{ln} \in \C$ for all $\textbf{k},\textbf{k}' \in \textbf{K}$ and $n\in \Z$. If we also define $c_{\textbf{k}\textbf{k}'}^{l}(x)=\sum_{n \in \Z} c_{\textbf{k}\textbf{k}'}^{ln} x^{-n-1}$, we have that $$F_{l}(x)(X,Y) = \sum_{\textbf{k},\textbf{k}' \in \textbf{K}} c_{\textbf{k}\textbf{k}'}^{l}(x)\, X^{\textbf{k}} Y^{\textbf{k}'}.$$

The first property to be proved is simply the result of writing the convergence axiom of $\Ydual(x)$ in $\textbf{V}$, using the fact that the ideals $\mathcal{M}_{r,I}$ form a basis of open neighborhoods of zero in $\C[[X_i,Y_i : i \in \mathcal{I}]]$.

On the other hand, the covacuum map $c: \textbf{V} \rightarrow \C$ may be computed as $c(X^{\textbf{k}})=\delta_{\textbf{k},0}$ for all $\textbf{k} \in \textbf{K}$. This allows us to calculate
\begin{align*}
(c \hotimes Id) \Ydual(x)(X_{l}) &= (c \hotimes Id) \left( \sum_{\textbf{k},\textbf{k}' \in \textbf{K}} c_{\textbf{k}\textbf{k}'}^{l}(x) X^{\textbf{k}} \hotimes X^{\textbf{k}'} \right)\\
&= \sum_{\textbf{k},\textbf{k}' \in \textbf{K}} c_{\textbf{k}\textbf{k}'}^{l}(x)\, c(X^{\textbf{k}}) \hotimes X^{\textbf{k}'}\\
&=\sum_{\textbf{k}' \in \textbf{K}}c_{0\textbf{k}'}^{l}(x)\,X^{\textbf{k}'}\\
&=\sum_{\textbf{k},\textbf{k}' \in \textbf{K}} c_{\textbf{k}\textbf{k}'}^{l}(x)\, Y^{\textbf{k}} Z^{\textbf{k}'} \Big\rvert_{\substack{Y=\,0\\Z=X}}.\\
&=F_{l}(x)(0,X).
\end{align*}

Therefore, the covacuum is a left counit in $\textbf{V}$ if and only if $F_{l}(x)(0,X)=X_{l}$ for all $l \in \mathcal{I}$, which is condition (\ref{left_identity_CFGL}). The cocreation axiom of the topological vertex coalgebra $\textbf{V}$ may be transformed into condition (\ref{right_identity_CFGL}) through a similar procedure. 

Lastly, we have

\begin{align*}
x&_0^{-1}\delta\left(\frac{x_1-x_2}{x_0}\right)(Id \hotimes \Ydual(x_2))\Ydual(x_1)(X_{l})\\
&=x_0^{-1}\delta\left(\frac{x_1-x_2}{x_0}\right)(Id \hotimes \Ydual(x_2))\sum_{\textbf{k},\textbf{k}' \in \textbf{K}} c_{\textbf{k}\textbf{k}'}^{l}(x_1)\, X^{\textbf{k}} \hotimes X^{\textbf{k}'}\\
&=x_0^{-1}\delta\left(\frac{x_1-x_2}{x_0}\right)\sum_{\textbf{k},\textbf{k}' \in \textbf{K}} c_{\textbf{k}\textbf{k}'}^{l}(x_1)\, X^{\textbf{k}} \hotimes  \prod_{i' \in \mathcal{I}}\Ydual(x_2)(X_{i'})^{k_{i'}'}\\
&=x_0^{-1}\delta\left(\frac{x_1-x_2}{x_0}\right)\sum_{\textbf{k},\textbf{k}' \in \textbf{K}} c_{\textbf{k}\textbf{k}'}^{l}(x_1)\, X^{\textbf{k}} \hotimes \prod_{i' \in \mathcal{I}}\left( \sum_{\textbf{t},\textbf{t}' \in \textbf{K}} c_{\textbf{t}\textbf{t}'}^{i'}(x_2)\, X^{\textbf{t}} \hotimes X^{\textbf{t}'} \right)^{k_{i'}'}\\
&=x_0^{-1}\delta\left(\frac{x_1-x_2}{x_0}\right)\sum_{\textbf{k},\textbf{k}' \in \textbf{K}} c_{\textbf{k}\textbf{k}'}^{l}(x_1)\, X^{\textbf{k}} \tilde{Y}^{\textbf{k}'} \Big\rvert_{\tilde{Y}=F(x_2)(Y,Z)}\\
&=x_0^{-1}\delta\left(\frac{x_1-x_2}{x_0}\right)F_{l}(x_1)(X,F(x_2)(Y,Z)),
\end{align*}
with $X=X \hotimes 1 \hotimes 1$, $Y=1 \hotimes X \hotimes 1$ and $Z= 1 \hotimes 1 \hotimes X$ in $\textbf{V} \hotimes \textbf{V} \hotimes \textbf{V}$. Performing similar calculations in the other two terms of the co-Jacobi identity in $\textbf{V}$ proves that both sides in (\ref{Jacobi_CFGL}) are well-defined and coincide. 
\end{proof}

\begin{definition}
A \textit{formal vertex law} (of dimension $|\mathcal{I}|$) is a set of formal power series $$F(x)(X,Y)\in (\C[[x,x^{-1},X_{i},Y_{i}: i \in \mathcal{I} ]])^{\mathcal{I}}$$ such that properties (\ref{convergence_vertex_formal_law}) through (\ref{Jacobi_CFGL}) hold.
\end{definition}

We would like to define the category of formal vertex laws so that the analogy with the classical equivalence between the categories of formal group laws and Lie algebras is preserved.

\begin{definition}
Let $F(x)(X,Y)$ and $G(x)(X,Y)$ be two formal vertex laws of dimensions $|\mathcal{I}|$ and $|\mathcal{J}|$. A \textit{homomorphism} between them is a set of formal power series $\alpha(X) \in \C[[X_i : i \in \mathcal{I}]]^{\mathcal{J}}$ whose constant term is zero, such that
\begin{equation}
\alpha(F(x)(X,Y))=G(x)(\alpha(X),\alpha(Y)).
\end{equation}
\end{definition}

With these homomorphisms, formal vertex laws form a category, which we shall denote as $\mathcal{FVL}$.

\begin{proposition}\label{formal_vertex_laws_and_vertex_bicoalgebras}
The category of formal vertex laws is antiequivalent to the category of vertex bicoalgebras whose subjacent linearly compact algebra structure is commutative and local.
\end{proposition}

\begin{proof}
First, it is well-known that any linearly compact algebra that is\linebreak commutative and local is isomorphic to a formal power series algebra $A=\C[[X_i : i \in \mathcal{I}]]$ for some $\mathcal{I}$. Thus any vertex bicoalgebra $V$ with those\linebreak properties must be of that form, and we may use Proposition \ref{def_formal_vertex_law} to assign a formal vertex law to $V$, which we shall denote as $F^{V}(x)(X,Y)$.

On the other hand, starting from a formal vertex law $F(x)(X,Y)$ of dimension $|\mathcal{I}|$, we can construct a vertex bicoalgebra structure over $V = \C [[X_i : i \in \mathcal{I}]]$ as follows: since axiom (\ref{convergence_vertex_formal_law}) translates to $F_i(x)(X,Y)$ being an element of the algebra $(V \hotimes V)[[x,x^{-1}]]^{conv}$, the universal property of the power series algebra $V$ implies that there exists a unique algebra map $\Ydual(x)$ from $V$ to that algebra such that its value at each $X_{i}$ coincides with $F_{i}(x)(X,Y)$. The covacuum map $c$ is also uniquely determined because there is only one algebra homomorphism $V \rightarrow \C$, due to $V$ being local.

Thus we have shown that for each formal vertex law of dimension $|\mathcal{I}|$ there exists a unique vertex bicoalgebra structure over $\textbf{V}=\C[[X_i : i \in \mathcal{I}]]$ such that (\ref{def_Fijm}) holds. We will now see that this correspondence extends to morphisms in their respective categories, which will give us the desired antiequivalence.

If $\psi : \tilde{V} \rightarrow V$ is a linear map between commutative local vertex bicoalgebras, we may assume that $V = \C [[X_i : i \in \mathcal{I}]]$ and $\tilde{V} = \C [[\tilde{X}_j : j \in \mathcal{J}]]$. Now we define $\alpha(X) \in (\C[[X_i : i \in \mathcal{I}]])^{\mathcal{J}}$ as $\alpha_{j}(X) = \psi(\tilde{X}_j)$ for all $j \in \mathcal{J}$. Let us write $\alpha_j(X)=\sum_{m \in \textbf{K}}\alpha_{j\textbf{m}}X^{\textbf{m}}$ with $\alpha_{j\textbf{m}}\in \C$ for all $j \in \mathcal{J}$ and $\textbf{m} \in \textbf{K}$. Then for all $j \in \mathcal{J}$ we have that

\begin{align*}
\alpha_j(F^{V}(x)(X,Y)) &= \sum_{m \in \textbf{K}}\alpha_{j\textbf{m}}X^{\textbf{m}}\Big|_{X=F^{V}(x)(X,Y)}\\
&=\sum_{m \in \textbf{K}}\alpha_{j\textbf{m}}\prod_{i \in \mathcal{I}}(\Ydual_V(x)(X_i))^{\textbf{m}_i}\\
&=\Ydual_{V}(x)(\alpha_j(X))\\
&=\Ydual_{V}(x)(\psi(\tilde{X}_j)).
\end{align*}

Similarly, 
\begin{align*}
F_j^{\tilde{V}}(x)(\alpha(X),\alpha(Y))&=\sum_{\tilde{\textbf{k}},\tilde{\textbf{k}}' \in \tilde{\textbf{K}}}c^{j}_{\tilde{\textbf{k}}\tilde{\textbf{k}}'}(x)\tilde{X}^{\tilde{\textbf{k}}} \hotimes\tilde{X}^{\tilde{\textbf{k}}'}\Big|_{\tilde{X}=\alpha(X)}\\
&=\sum_{\tilde{\textbf{k}},\tilde{\textbf{k}}' \in \tilde{\textbf{K}}} c^{j}_{\tilde{\textbf{k}}\tilde{\textbf{k}}'}(x)\prod_{j \in \mathcal{J}} \psi(\tilde{X}_j)^{k_j} \hotimes \prod_{j' \in \mathcal{J}} \psi(\tilde{X}_{j'})^{k'_{j'}}\\
&=(\psi \hotimes \psi) \sum_{\tilde{\textbf{k}},\tilde{\textbf{k}}' \in \tilde{\textbf{K}}}c^{j}_{\tilde{\textbf{k}}\tilde{\textbf{k}}'}(x)\tilde{X}^{\tilde{\textbf{k}}} \hotimes\tilde{X}^{\tilde{\textbf{k}}'}\\
&=(\psi \hotimes \psi)\Ydual_{\tilde{V}}(x)(\tilde{X}_j).
\end{align*}

Therefore, $\psi$ is a vertex bicoalgebra homomorphism if and only if $\alpha(X)$ is a formal vertex law homomorphism.
\end{proof}

Putting together this result, a restricted version of the Cartier duality proved in the last section (using the fact that the concepts of locality for linearly compact algebras and connectedness for coalgebras are dual to each other) and Proposition \ref{equiv_conformal_and_vertex_bialg}, we immediately arrive to the following important statement.

\begin{corollary}\label{correspondence_FVL_Conformal}
The category of formal vertex laws is equivalent to the\linebreak category of Lie conformal algebras.
\end{corollary}

The results obtained in this section show that formal vertex laws are completely analogous to formal group laws in classical Lie theory. Since formal group laws may be used to define Lie groups under suitable convergence conditions, it is reasonable to expect that we may use formal vertex laws in a similar fashion, obtaining an object which would be the conformal analogous to a Lie group. We are currently working towards this goal, and we have found promising relations between these results and Beilinson and Drinfeld's notion of \textit{factorization space} (also known as chiral monoid \cite{Beilinson_Drinfeld}), which shall appear in future publications.\\

\end{document}